\documentclass[11pt,letterpaper]{article}
\usepackage{amsthm,amsmath,amssymb,amsfonts} 
\usepackage{epsfig} 
\usepackage{latexsym,nicefrac,bbm}
\usepackage{xspace}
\usepackage{color,fancybox,graphicx,subfigure,fullpage}
\usepackage[top=1in, bottom=1in, left=1in, right=1in]{geometry}
\usepackage{tabularx}
\usepackage{hyperref} 
\usepackage{pdfsync}
\usepackage{multicol}
\usepackage{cite,cleveref}

\renewcommand{\epsilon}{\varepsilon}

\newcommand{\argmin}{\operatornamewithlimits{argmin}}

\newcommand{\eps}{\varepsilon}

\usepackage{epsfig}
\usepackage{verbatim}

\theoremstyle{definition}

\usepackage{algorithmicx}
\usepackage{algorithm,caption}
\usepackage{algpseudocode}

\usepackage{mhequ}
\def \be{\begin{equs}}
\def \ee{\end{equs}}

\newtheorem{theorem}{Theorem}[section]
\newtheorem{lemma}[theorem]{Lemma}
\newtheorem{definition}[theorem]{Definition}

\newtheorem{corollary}[theorem]{Corollary}

\newtheorem{proposition}[theorem]{Proposition}

\newtheorem*{theorem*}{Theorem}

\crefname{theorem}{Theorem}{Theorems}
\crefname{observation}{Observation}{Observations}
\crefname{proposition}{Proposition}{Propositions}
\crefname{claim}{Claim}{Claims}
\crefname{condition}{Condition}{Conditions}
\crefname{example}{Example}{Examples}
\crefname{fact}{Fact}{Facts}
\crefname{lemma}{Lemma}{Lemmas}
\crefname{corollary}{Corollary}{Corollaries}
\crefname{definition}{Definition}{Definitions}
\crefname{remark}{Remark}{Remarks}
\crefname{construction}{Construction}{Constructions}

\title{\bf On Geodesically Convex Formulations for the Brascamp-Lieb Constant}
\author{Nisheeth K. Vishnoi and Ozan Y\i ld\i z \\ {\'{E}cole Polytechnique F\'{e}d\'{e}rale de Lausanne (EPFL), Switzerland}}

\usepackage{dsfont}
\usepackage{mathtools}
\usepackage{graphicx}
\newcommand{\B}{\mathit{B}}
\newcommand{\p}{\mathit{p}}
\newcommand{\R}{\mathbb{R}}
\newcommand{\Z}{\mathbb{Z}}

\newcommand{\Sn}{\mathbb{S}}

\newcommand{\BL}{\mathrm{BL}}

\newcommand{\tr}{\mathrm{Tr}}

\DeclarePairedDelimiterX{\Set}[2]\{\}{%
  \, #1 \;\delimsize\vert\; #2 \,
}

\begin{document}

\maketitle

\begin{abstract}
We consider  two  non-convex formulations for computing the optimal constant in the Brascamp-Lieb inequality corresponding to a given datum, and show that they are geodesically log-concave on the manifold of positive definite matrices endowed with the Riemannian metric corresponding to the Hessian of the log-determinant function.
The first formulation is present in the work of  Lieb \cite{lieb1990gaussian} and the second is inspired by the work of Bennett {\em et al.} \cite{bennett2008brascamp}.
Recent works of Garg {\em et al.} \cite{garg2016algorithmic} and Allen-Zhu {\em et al.} \cite{OS2018} also imply a geodesically log-concave formulation of the Brascamp-Lieb constant through a reduction to the operator scaling problem.
However, the dimension of the arising optimization problem in their reduction depends exponentially on the number of bits needed to describe the Brascamp-Lieb datum. 
The formulations presented here have dimensions that are polynomial in the bit complexity of the input datum.

\end{abstract}  
  \thispagestyle{empty}

\section{Introduction}

\paragraph{The Brascamp-Lieb Inequality.}

Brascamp and Lieb  \cite{brascamp1976best} presented a class of inequalities that generalize many well-known inequalities  and, as a consequence, have played an important role in various mathematical disciplines.
Formally, they presented the following class of inequalities where each inequality is described by a ``datum'', referred to as the Brascamp-Lieb datum.

\begin{definition}[The Brascamp-Lieb Inequality, Datum, Constant]
Let $n$, $m$, and $(n_j)_{j\in[m]}$ be positive integers and $\p:=(p_j)_{j\in[m]}$ be non-negative real numbers.
Let $\B:=(B_j)_{j\in[m]}$ be an $m$-tuple of linear transformations where $B_j$ is a surjective linear transformation from $\R^{n}$ to $\R^{n_j}$.
The corresponding  Brascamp-Lieb {\em datum} denoted by $(\B,\p)$.
The Brascamp-Lieb inequality states that for each Brascamp-Lieb datum $(\B,\p)$ there exists a constant $C(\B,\p)$ not necessarily finite, such that for any selection of real-valued, non-negative, Lebesgue measurable functions $f_j$ where $f_j:\R^{n_j}\rightarrow \R$,
\begin{equ}
\int_{x\in\R^n} \left(\prod_{j\in[m]} f_j(B_j x)^{p_j}\right)dx
\leq
C(\B,\p) \prod_{j\in[m]}\left(\int_{x\in \R^{n_j}} f_j(x) dx \right)^{p_j}.
\label{eq:BLinequality}
\end{equ}
The smallest constant that satisfies~\eqref{eq:BLinequality} for any choice of $f:=(f_j)_{j\in[m]}$ satisfying the properties mentioned above is called the Brascamp-Lieb {\em constant} and denoted by $\BL(\B,\p)$. 
A Brascamp-Lieb datum $(\B,\p)$ is called \emph{feasible} if $\BL(\B,\p)$ is finite, otherwise, it is called \emph{infeasible}.
For a given $m$-tuple $B$, the set of real vectors $\p$ such that $(\B,\p)$ is feasible is denoted by $P_{\B}$.
\label{def:BL}
\end{definition}

\noindent
Applications of the Brascamp-Lieb inequality extend beyond  functional analysis and appear in  convex geometry~\cite{ball1989volumes}, information theory~\cite{carlen2009subadditivity},\cite{LiuCCV16},\cite{LiuCCV17}, machine learning~\cite{hardt2013algorithms}, and theoretical computer science \cite{DSW14, DvirGOS16}.

\paragraph{Mathematical Aspects of the Brascamp-Lieb Inequality.}

\noindent
A Brascamp-Lieb inequality is non-trivial only when $(\B,\p)$ is a feasible Brascamp-Lieb datum.
Therefore, it is of interest to  characterize feasible Brascamp-Lieb data and compute the corresponding Brascamp-Lieb constant.
Lieb~\cite{lieb1990gaussian} showed that one needs to consider only Gaussian functions as inputs for~\eqref{eq:BLinequality}.
This result suggests the following characterization of the Brascamp-Lieb constant as an optimization problem.
For a positive integer $k$, let $\Sn^{k}_+$ be the space of real-valued, symmetric, positive semi-definite (PSD) matrices of dimension $k\times k$. 
\begin{theorem}[Gaussian maximizers \cite{lieb1990gaussian}]
Let $(\B,\p)$ be a Brascamp-Lieb  datum with $B_j\in\R^{n_j\times n}$ for each $j\in[m]$.
Let $A:=(A_j)_{j\in[m]}$ with $A_j\in\Sn^{n_j}_+$, and consider the function
\begin{equ}[eq:BLconstant]
\BL(\B,\p;A) := \left(\frac{\prod\limits_{j\in[m]}\det(A_j)^{p_j}}{\det(\sum\limits_{j\in[m]}p_jB_j^\top A_jB_j)}\right)^{1/2}.
\end{equ}
Then, the Brascamp-Lieb constant for $(\B,\p)$, $\BL(\B,\p)$ is equal to   $\sup\limits_{A\in \bigtimes\limits_{j\in[m]} \Sn^{n_j}_+} \BL(\B,\p;A)$.
\label{thm:BLconstant}
\end{theorem}

%
\noindent
Bennett {\em et al.}~\cite{bennett2008brascamp} analyzed the following necessary and sufficient conditions for the feasibility of a Brascamp-Lieb datum by extending  Lieb's work.

\begin{theorem}[Feasibility of Brascamp-Lieb Datum \cite{bennett2008brascamp}, Theorem 1.15]
Let $(\B,\p)$ be a Brascamp-Lieb  datum with $B_j\in\R^{n_j\times n}$ for each $j\in[m]$. 
Then, $(\B,\p)$ is feasible if and only if following conditions hold:
\begin{enumerate}
	\item $n=\sum_{j\in[m]}p_jn_j$, and
	\item $\dim(V)\leq \sum_{j\in[m]} p_j \dim(B_j V)$ for any subspace $V$ of $\R^n$.
\end{enumerate}
\label{thm:feasibility}
\end{theorem}

\noindent
\cref{thm:feasibility} introduces infinitely many linear restrictions on $\p$ while $V$ varies over different subspaces of $\R^n$.
However, there are only finitely many different linear restrictions as $\dim(B_j V)$ can only take integer values from $[n_j]$.
Consequently, this theorem implies that $P_{\B}$ is a convex set and, in particular a polytope. 
It is referred to as the Brascamp-Lieb polytope \cite{valdimarsson2010brascamp}.
Some of the above inequality constraints are tight for any $\p\in P_{\B}$ such as the inequality constraints induced by $\R^n$ and the trivial subspace, while others can be strict for some $p\in P_{\B}$.
If $\p$ lies on the boundary of $P_{\B}$, then there should be some non-trivial subspaces $V$ such that the induced inequality constraints are tight for $\p$.
This  leads to the definition of critical subspaces and simple Brascamp-Lieb datums.
%

\begin{definition}[Critical Subspaces and Simple Brascamp-Lieb Data]
Let $(\B,\p)$ be a feasible Brascamp-Lieb  datum with $B_j\in\R^{n_j\times n}$ for each $j\in[m]$. 
Then, a subspace $V$ of $\R^n$ is called \emph{critical} if 
\begin{equ}
\dim(V)=\sum_{j\in[m]}p_j\dim(B_jV).
\end{equ}
$(\B,\p)$ is called \emph{simple} if there is no non-trivial proper subspace of $\R^n$ which is critical.
\end{definition}

\noindent
For a fixed $B$, simple Brascamp-Lieb data correspond to points $p$ that lie in the relative interior of the Brascamp-Lieb polytope $P_{\B}$.
One important property of simple Brascamp-Lieb data is that there exists a maximizer for $\BL(\B,\p;A)$.
This was proved  by Bennett {\em et al.}~\cite{bennett2008brascamp} by analyzing Lieb's formulation~\eqref{eq:BLconstant}.
This analysis also leads to a characterization of maximizers of $\BL(\B,\p;A)$.

\begin{theorem}[Characterization of Maximizers\cite{bennett2008brascamp}, Theorem 7.13]
Let $(\B,\p)$ be a Brascamp-Lieb  datum with $B_j\in\R^{n_j\times n}$  and  $p_j>0$ for all $j\in[m]$.
Let $A:=(A_j)_{j\in[m]}$ be an $m$-tuple of positive semidefinite matrices with $A_j\in\R^{n_j\times n_j}$ and let $M:=\sum_{j\in[m]}p_j B_j^\top A_jB_j$.
Then, the following statements are equivalent,
\begin{enumerate}
	\item $A$ is a global maximizer for $\BL(\B,\p;A)$ as in \cref{eq:BLconstant}.
	\item $A$ is a local maximizer for $\BL(\B,\p;A)$.
	\item $M$ is invertible, and $A_j^{-1} = B_j M^{-1} B_j^\top$ for each $j\in[m]$.
\end{enumerate}
Furthermore, the global maximizer $A$ for $\BL(\B,\p;A)$ exists and is unique up to scalar if and only if $(\B,\p)$ is simple. 
\label{thm:EquivalencetoGeometric}
\end{theorem}

\paragraph{Computational Aspects of the Brascamp-Lieb Inequality.}
One of the computational questions concerning the Brascamp-Lieb inequality is: Given a Brascamp-Lieb datum $(B,p)$, can we compute $\BL(\B,\p)$ in polynomial time?
Since computing $\BL(\B,\p)$ exactly may not be possible due to the fact that this number may not be rational even if the datum $(\B,\p)$ is, one seeks an arbitrarily good approximation.
Formally, given the entries of $B$ and $p$ in binary, and an $\eps>0$, compute a number $Z$ such that 
$$ \BL(\B,\p) \leq Z \leq  (1+\eps) \; \BL(\B,\p)$$
in time that is polynomial in the combined bit lengths of $B$ and $p$ and $\log 1/\eps.$

There are a few obstacles to this problem: (1) Checking if a given Brascamp-Lieb datum is feasible is not known to be in {\bf P}. 
(2)  The formulation of the Brascamp-Lieb constant by Lieb~\cite{lieb1990gaussian} as in \cref{eq:BLconstant} is neither concave nor log-concave.
Thus,  techniques developed in the context of linear and convex optimization do not seem to be directly applicable.

An important step towards the computability of  $\BL(\B,\p)$ was taken recently by Garg {\em et al.} \cite{garg2016algorithmic} where they presented a pseudo-polynomial time algorithm for (1) and a pseudo-polynomial time algorithm to compute  $\BL(\B,\p)$. 
The running time of their algorithms depended polynomially on the {\em magnitude} of the denominators in the components of $p$ rather than the number of bits required to represent them.
Interestingly, Garg {\em et al.}  presented a reduction of the problem of computing $\BL(\B,\p)$ to the problem of computing the ``capacity'' in an  ``operator scaling'' problem considered by Gurvits \cite{GURVITS2004448}.
Roughly, in the operator scaling problem, given a representation of a linear mapping from PSD matrices to PSD matrices, the goal is to compute the minimum ``distortion'' of this mapping (see \cref{def:os}).
The operator scaling problem is also not a concave or log-concave optimization problem.
However,  very recently, operator scaling was shown to be ``geodesically'' log-concave by Allen-Zhu {\em et al.}~\cite{OS2018}.

Geodesic convexity is an extension of convexity along straight lines in Euclidean spaces to geodesics in Riemannian manifolds. 
Since all the problems mentioned so far are defined on positive definite matrices, the natural  manifold to consider is the space of positive definite matrices with a particular Riemannian metric: the Hessian of the $- \log \det$ function; see \cref{sec:geodesic}.
Geodesics are analogues of straight lines on a manifold and, 
roughly, a function $f$ on a Riemannian manifold is said to be geodesically convex if the average of its values at the two end points of any geodesic is at least its value at its mid-point.

The reduction of Garg {\em et al.}~\cite{garg2016algorithmic}, thus, leads to a geodesically log-concave formulation to compute $\BL(\B,\p)$.
However, their construction does not lead to an optimization  problem whose dimension is polynomial in the input bit length as the size of constructed positive operator in the operator scaling problem depends exponentially on the bit lengths of the entries of $p$. 
More precisely, if $p_j=\nicefrac{c_j}{c}$ for integers $(c_j)_{j\in[m]}$ and $c$, then the aforementioned construction results in operators over  $\Sn^{nc}_{++}$.

Besides this line of work, a polynomial time algorithm to compute  $\BL(\B,\p)$, in the special case of rank one ($n_j=1$ for all $j$), was presented by Straszak and Vishnoi  \cite{straszak2017entropy} through a connection with computing maximum entropy distributions.

\paragraph{Our Contribution.}
Our first result is  that  Lieb's formulation presented in \cref{thm:BLconstant} is jointly geodesically log-concave with respect to inputs $(A_j)_{j\in[m]}$.
\begin{theorem}[Geodesic Log-Concavity of Lieb's Formulation]\label{thm:Lieb}
Let $(\B,\p)$ be a feasible Brascamp-Lieb  datum with $B_j\in\R^{n_j\times n}$ for each $j\in[m]$.
Then, $\BL(\B,\p;A)$ is jointly geodesically log-concave with respect to $A:=(A_j)_{j\in[m]}$ where $\BL(\B,\,\p;\,A)$ is defined as~\eqref{eq:BLconstant}.
\end{theorem}
\noindent
This formulation leads to a geodesically convex optimization problem on $\bigtimes_{j\in[m]}  \Sn_{++}^{n_j}$ that captures the Brascamp-Lieb constant.

Subsequently, we present a  modified version of Lieb's formulation by combining it with observations made by Bennett {\em et al.}~\cite{bennett2008brascamp} about maximizers of $\BL(\B,\,\p;\,A)$; see 
  \cref{thm:EquivalencetoGeometric}.
\cite{bennett2008brascamp} showed that if $A=(A_j)_{j\in[m]}$ is a maximizer to \cref{eq:BLconstant}, then
$A_j = (B_j M^{-1} B_j^\top)^{-1}$
for each $j\in[m]$, where
$M := \sum_{j\in[m]} p_j B_j^\top A_j B_j.$
Thus, we can write each $A_j$ as a function of $M$ and obtain, $2\log(\BL(\B,\p;A(M)))$ equals
\begin{equation}\label{eq:ben}
\sum_{j\in[m]} p_j \log\det( (B_j M^{-1} B_j^\top)^{-1}) - \log\det\left(\sum_{j\in[m]} p_j B_j^\top (B_j M^{-1} B_j^\top)^{-1} B_j\right). 
\end{equation}
One can show that the expressions
\begin{equ}
\log\det( (B_j M^{-1} B_j^\top)^{-1})
\end{equ}
for each $j\in[m]$ and
\begin{equ}
\log\det\left(\sum_{j\in[m]} p_j B_j^\top (B_j M^{-1} B_j^\top)^{-1} B_j\right)
\end{equ}
are geodesically concave with respect to $M$.
However, the expression in \cref{eq:ben} being a difference, is not geodesically concave with respect to $M$ in general.
However, if $A$ is a global maximizer of $\BL(\B,\p;A)$, then we also have that
\begin{equ}
M = \sum_{j\in[m]} p_j B_j^\top (B_j M^{-1} B_j^\top)^{-1} B_j.
\end{equ}
Combining these two observations, we obtain the following new geodesically concave optimization problem for computing the Brascamp-Lieb constant.

\begin{theorem}[A  Geodesically Log-Concave Formulation of the Brascamp-Lieb Constant]
Let $(\B,\p)$ be a feasible Brascamp-Lieb  datum with $B_j\in\R^{n_j\times n}$ for each $j\in[m]$. 
Let  $F_{\B,\,\p}(X):\mathds{S}_{++}^n\rightarrow\R$ be defined as follows,
\begin{equ}[eq:F]
F_{\B,\p}(X) := \log\det(X) - \sum_{j\in[m]} p_j\log\det(B_j X B_j^\top).
\end{equ}
Then, $F_{\B,\p}$ is geodesically concave.
Furthermore, if $(\B,\p)$ is simple, then $\sup\limits_{X\in \Sn^n_{++}}
F_{\B,\p}(X)$ is attained.
If $X^\star$ is a maximizer of $F_{\B,\p}$, then $\exp(\frac{1}{2}F_{\B,\p}(X^\star))=\BL(\B,\p)$ and $A^\star=((B_j X^\star B_j^\top)^{-1})_{j\in[m]}$ maximizes $\BL(\B,\p;\,A^\star)$.
\label{thm:new_formulation}
\end{theorem}

\section{The Positive Definite Cone, its Riemannian Geometry, and Geodesic Convexity}
\label{sec:geodesic}
\paragraph{The Metric.} Consider the set of positive definite matrices $\Sn^d_{++}$ as a subset of $\mathbb{R}^{d \times d}$ with the inner product $\langle X,Y\rangle := \tr(X^\top Y)$ for $X,Y \in \mathbb{R}^{d \times d}$.
At any point $X \in \Sn^d_{++}$, the tangent space consists of  all $d \times d$ real symmetric matrices. 
There is a natural metric $g$ on this set that gives it a Riemannian structure:
For $X \in \Sn^d_{++}$ and two symmetric matrices $\nu,\xi$  
\begin{equ}\label{eq:metric}
g_X(\nu, \xi):=\langle X^{-1} \nu,X^{-1} \xi\rangle.
\end{equ}
It is an exercise in differentiation to check that this metric arises as the Hessian of the following function $\varphi:\Sn^d_{++}\to \mathbb{R}$:
$$ \varphi(X):= -\log \det X.$$
Hence,  $\Sn^d_{++}$ endowed with the metric $g$ is not only a Riemannian manifold, but a Hessian manifold \cite{Shima07HessianBook}. 
The study of this metric on $\Sn^d_{++}$ goes back at least to Siegel \cite{siegel1943symplectic}; see also the book of Bhatia \cite{bhatia2009positive}.

\noindent
\paragraph{Geodesics on $\Sn^d_{++}$.} If $X,Y \in \Sn^d_{++}$ and $\gamma:[0,1]\to\Sn^d_{++}$ is a smooth curve between $X$ and $Y$, then the arc-length of $\gamma$ is given by the (action) integral
\begin{equ}[eq:length]
L(\gamma):= \int_0^1 g_{\gamma(t)}\left(\frac{d\gamma(t)}{dt},\frac{d\gamma(t)}{dt}\right)dt.
\end{equ}
The geodesic between $X$ and $Y$ is the unique smooth curve between $X$ and $Y$ with the smallest arc-length.
The Euler-Lagrange equations describing a geodesic can be obtained by a variational approach to $L(\gamma)$. 
The following theorem asserts that between any two points in $\Sn^d_{++}$, there is a 
geodesic that connects them.
In other words, $\Sn^d_{++}$ is a geodesically convex set.
Moreover, there exists a closed form expression for the geodesic between two points, a formula that is useful for calculations.

\begin{theorem}[Geodesics on $\Sn^d_{++}$ \cite{bhatia2009positive}, Theorem 6.1.6]\label{thm:geodesics_pd}
For $X, Y \in \Sn^d_{++}$,
the exists a unique geodesic between $X$ and $Y$, and this geodesic is parametrized by the following equation: 
\begin{equ}
X\#_t Y := X^{1/2} (X^{-1/2} Y X^{-1/2})^{t} X^{1/2}
\end{equ}
for $t\in[0,1]$.
\end{theorem}

\paragraph{Geodesic Convexity.} 
One definition of convexity of a function $f$ in a Euclidean space is that the average of the function at the end points of each line in the domain is at-least  the value of the function at the average point on the line.
Geodesic convexity is a natural extension of this notion of convexity from Euclidean spaces to Riemannian manifolds that are geodesically convex.
A set in the manifold is said to be geodesically convex if for every pair of points in the set, the geodesic combining these points lies entirely in the set.

\begin{definition}[Geodesically Convex Sets]
A set $S\subseteq \Sn^{d}_{++}$ is called geodesically convex if for any $X,Y\in S$ and $t\in[0,1]$, $X\#_tY\in S$.
\end{definition}
\noindent
A function defined on a geodesically convex set is said to be geodesically convex if the average of the function at the end points of any geodesic in the domain is at-least  the value of the function at the average point on the geodesic.

\begin{definition}[Geodesically Convex Functions]
Let $S\subseteq \Sn^{d}_{++}$ be a geodesically convex set.
A function $f:S\rightarrow \R$ is called geodesically convex if for any $X,Y\in\Sn^d_{++}$ and $t\in[0,1]$,
\begin{equ}
f(X\#_t Y) \leq (1-t) f(X) + t f(Y).
\end{equ}
$f$ is called geodesically concave if $-f$ is geodesically convex.
\end{definition}

\noindent
An important point regarding geodesic convexity is that a non-convex function might be geodesically convex or vice-versa.
In general, one cannot convert a geodesically convex function to a convex function by a change of variables.
A well-known example for this is the $\log\det(X)$ function whose concavity a classical result from the matrix calculus.
On the other hand, a folklore result is that $\log\det(X)$ is both geodesically convex and geodesically concave  on the space of positive definite matrices with the metric \cref{eq:metric}. 

\begin{proposition}[Geodesic Linearity of $\log\det$]
The $\log\det(X)$ function is geodesically linear, i.e, it is both geodesically convex and geodesically concave over $\Sn^n_{++}$.
\label{thm:det-linear}
\end{proposition}

\begin{proof}
Let $X, Y\in\Sn^n_{++}$ and $t\in[0,1]$. 
Then,
\begin{equ}
\log\det(X\#_tY) \stackrel{\cref{thm:geodesics_pd}}{=} \log\det(X^{1/2} (X^{-1/2}Y X^{-1/2})^t X^{1/2}) = (1-t)\log\det(X) + t\log\det(Y).
\end{equ}
Therefore, $\log\det(X)$ is a geodesically linear function over positive definite cone with respect to the metric in~\cref{eq:metric}.
\end{proof}

\noindent
Henceforth, when we mention geodesic convexity it is with respect to the metric in \cref{eq:metric}.
Geodesically convex functions share some properties with usual convex functions.
One such property is the relation between local and global minimizers.

\begin{theorem}[Minimizers of Geodesically Convex Functions~\cite{Rapcsak1997}, Theorem 6.1.1]
Let $S\subseteq \Sn^d_{++}$ be a geodesically convex set and $f:S\to\R$ be a geodesically convex function.
Then, any local minimum point of $f$ is also a global minimum of $f$.
More precisely, if $x^\star:=\argmin_{x\in O} f(x)$ for some open geodesically convex subset $O$ of $S$, then $f(x^\star)=\inf_{x\in S} f(x)$.
\label{thm:gc-minimizers}
\end{theorem}

\paragraph{Geometric Mean of Matrices and Linear Maps.}
While  the function $\log\det(P)$ is geodesically linear, our proof of \cref{thm:Lieb} relies on 
the geodesic convexity of $\log\det(\sum_{j \in [m]} p_j B_j^\top A_j B_j)$.
A simple but important observation is that, if $(\B,\p)$ is feasible, then  $p_j B_j^\top A_j B_j$ is a strictly positive linear map for each $j$ as proved below.
\begin{lemma}
Let $(\B,\,\p)$ be a feasible Brascamp-Lieb datum with $B_j\in\R^{n_j\times n}$ for each $j\in[m]$.
Then, $\Phi_j(X):=B_j^\top X B_j$ is a strictly positive linear map for each $j\in[m]$.
\label{lem:feasible-plm}
\end{lemma}

\begin{proof}
Let us assume that for some $j_0\in[m]$, $\Phi_{j_0}(X)$ is not strictly positive linear map.
Then, there exists $X_0\in\Sn^{n}_{++}$ such that $\Phi_{j_0}(X_0)$ is not positive definite.
Thus, there exists $v\in\R^{n_{j_0}}$ such that $v^\top \Phi_{j_0}(X_0) v \leq 0$.
Equivalently, $(B_{j_0} v)^\top X_0 (B_{j_0} v)\leq 0$.
Since $X_0$ is positive definite, we get $B_{j_0} v = 0$.
Hence, $v^\top B_{j_0}^\top B_{j_0} v = 0$.
Consequently, rank of $B_{j_0}$ is at most $n_j-1$ and $\dim(B_j \R^n)<n_j$.
Therefore, 
\begin{equ}
n = \dim(\R^n) \leq \sum_{j\in[m]} p_j \dim(B_j \R^n) < \sum_{j\in[m]} p_j n_j = n,
\end{equ}  
by~\cref{thm:feasibility},
 a contradiction.
Consequently, for any $j\in[m]$, $\Phi_j(X):=B_j^\top X B_j$ is strictly positive linear whenever $(\B,\p)$ is feasible.
\end{proof}

\noindent
The joint geodesic convexity of $\log\det(\sum_{j \in [m]} p_j B_j^\top A_j B_j)$ follows from a more general observation (that we prove) that asserts that  if $\Phi_j$s  are strictly positive linear maps from $\Sn^{n_j}_+$ to $\Sn^{n}_+$, then 
$\log\det(\sum_{j \in [m]} \Phi_j(A_j))$ is geodescially convex.
Sra and Hosseini~\cite{sra2015conic} observed this when $m=1$.
Their result follows from a result of Ando~\cite{Ando1979} about  ``geometric means'' that is also important for us and  we explain it next.

The geometric mean of two matrices was introduced by Pusz and Woronowicz~\cite{Pusz1975}.
If $P,Q\in\Sn^d_{++}$, then the geometric mean of $P$ and $Q$ is defined as
\begin{equ}[eq:gm]
P\#_{1/2}Q=P^{1/2}(P^{-1/2}QP^{-1/2})^{1/2}P^{1/2}.
\end{equ}
By abuse of notation, we drop $\nicefrac{1}{2}$ and denote geometric mean by $P\#Q$.
Recall that, the geodesic convexity of a function $f:\Sn^{d}_{++}\to\R$ is equivalent to for any $P,Q\in\Sn^{d}_{++}$ and $t\in[0,\,1]$
\begin{equ}
f(P\#_t Q) \leq (1-t) f(P) + t f(Q).
\end{equ}
If $f$ is continuous, then the geodesic convexity of $f$ can be deduced from the following: 
\begin{equ}
\forall P,Q\in\Sn^{d}_{++}, \; \; f(P\#Q)\leq \frac{1}{2}f(P)+\frac{1}{2}f(Q).
\end{equ}

\noindent
 Ando proved the following result about the effect of a strictly positive linear map on the geometric mean of two matrices.

\begin{theorem}[Effect of a Linear Map over Geometric Mean~\cite{Ando1979}, Theorem 3]
Let $\Phi:\Sn^{d}_{+}\rightarrow \Sn^{d'}_{+}$ be a strictly positive linear map. If $P,Q\in\Sn^{d}_{++}$, then 
\begin{equ}[eq:lm-gm-ineq]
\Phi(P\#Q) \preceq \Phi(P)\#\Phi(Q). 
\end{equ}
\label{thm:plm-on-gm}
\end{theorem}

%
\noindent
The monotonicity of logdet ($P\preceq Q$ implies $\log\det(P)\leq \log\det(Q)$) 
and the multiplicativity of the determinant, combined with ~\cref{thm:plm-on-gm}, imply the following result.

\begin{corollary}[Geodesic Convexity of the Logarithm of Linear Maps\cite{sra2015conic}, Corollary 12]
If $\Phi:\Sn^{d}_{+}\to\Sn^{d'}_{+}$ is a strictly positive linear map, then $\log\det(\Phi(P))$ is geodesically convex.
\label{thm:log-lm-gc}
\end{corollary}

%
\noindent
While the proof of~\cref{thm:new_formulation} uses~\cref{thm:log-lm-gc}, it is not enough for the proof of~\cref{thm:Lieb}. 
Instead of geodesic convexity of $\log\det(\Phi(P))$, the joint geodesic convexity of $\log\det(\sum_{j\in[m]} \Phi_j(P_j))$ is needed where $\Phi_j:\Sn^{n_j}_{+}\to\Sn^{n}_{+}$ is a linear map for each $j\in[m]$.
We conclude this section with  the following two results on a maximal characterization of geometric mean and the effect of positive linear maps on positive definiteness of block diagonal matrices. 
%

\begin{theorem}[Maximal Characterization of the Geometric Mean~\cite{bhatia2009positive}, Theorem 4.1.1]
Let $P,Q\in\Sn^{d}_{++}$.
The geometric mean of $P$ and $Q$ can be characterized as follows,
\begin{equ}
P\#Q = \max\Set*{Y\in\Sn^{d}_{++}}{\begin{bmatrix} P & Y\\Y & Q\end{bmatrix}\succeq 0},
\end{equ}
where the maximal element is  with respect to Loewner partial order.
\label{thm:gm-maximal}
\end{theorem}

\begin{proposition}[Effect of Positive Linear Maps~\cite{bhatia2009positive}, Exercise 3.2.2]
Let $\Phi:\Sn^{d}_{+}\to\Sn^{d'}_{+}$ be a strictly positive linear map and $P,Q,R\in\Sn^{d}_{+}$.
If
\begin{equ}
\begin{bmatrix}
P & R\\
R & Q
\end{bmatrix}
\succeq 0,
\text{ then }
\begin{bmatrix}
\Phi(P) & \Phi(R)\\ 
\Phi(R) & \Phi(Q)
\end{bmatrix}
\succeq 0.
\end{equ}
\label{thm:plm-effect}
\end{proposition}

\section{A Geodesically Convex Formulation from Operator Scaling} 
\label{sec:os}
In this section, we describe the operator scaling problem and the reduction of Garg {\em et al.} \cite{garg2016algorithmic} from the computation of the Brascasmp-Lieb constant to the computation of the ``capacity'' of a positive operator.

\paragraph{The Operator Scaling Problem and its Geodesic Convexity.} In the operator scaling problem~\cite{GURVITS2004448}, one is given a linear operator $T(X):=\sum_{j\in[m]} T_j^\top X T_j$ through the tuple of matrices $T_j$s and the goal is to find square matrices $L$ and $R$ such that 
\begin{equ}[eq:scaling]
\sum_{j\in[m]} \hat{T}_j^\top \hat{T}_j = I\qquad\text{ and }\qquad\sum_{j\in[m]} \hat{T}_j \hat{T}_j^\top = I,
\end{equ}
where $\hat{T}_j := LT_jR$.
The matrices $L$ and $R$ can be computed by solving the following optimization problem. 
\begin{definition}[Operator Capacity]\label{def:os}
Let $T:\Sn^{d}_{++}\to\Sn^{d'}_{++}$ be a linear operator, then the capacity of $T$ is
\begin{equ}[eq:capacity]
\mathrm{cap}(T) := \inf_{\det(X)=1} \det\left(\frac{d}{d'}T(X)\right).
\end{equ}
\end{definition}

\noindent
In particular, if $X_T^\star$ is a minimizer of~\eqref{eq:capacity} and $Y_T^\star = T(X_T^\star)^{-1}$, then~\eqref{eq:scaling} holds if we let $L:=(Y_T^\star)^{1/2}$ and $R:=(X_T^\star)^{1/2}$; see ~\cite{GURVITS2004448} for details.
Allen-Zhu {\em et al.} \cite{OS2018} proved that the log of the operator capacity function is convex.

\begin{theorem}[Operator Capacity is Geodesically Convex \cite{OS2018}]
If $T$ is a positive operator, then $\log\mathrm{cap}(X)$ is a geodesically convex.
\end{theorem}

\paragraph{The Reduction.} Let $(\B,\p)$ be a Brascamp-Lieb  datum with $B_j\in\R^{n_j\times n}$ for each $j\in[m]$. 
Garg {\em et al.} \cite{garg2016algorithmic} proved that if the exponent $p=(p_j)_{j\in[m]}$ is a rational vector, then one can construct an operator scaling problem from $(\B,\p)$.
Let $p_j= \nicefrac{c_j}{c}$ where $c_j$s are non-negative integers, and $c$ is a positive integer, the common denominator for all the $p_j$s.
 Their reduction, outlined  below, results in an operator $T_{\B,\p}:\Sn^{nc}_{++}\to\Sn^{n}_{++}$ with the property that $\mathrm{cap}(T_{\B,\p})=\nicefrac{1}{\BL(\B,\p)^2}$. 

The operator $T_{\B,\p}$ is constructed with $c_j$ copies of the matrix $B_j$ for each $j\in[m]$.
In order to easily refer these copies, let us define $m':=\sum_{j\in[m]} c_j$, and the function $\delta:[m']\to[m]$.
$\delta(i)$ is defined as the integer $j$ such that,
\begin{equ}
\sum_{k<j} c_k < i \leq \sum_{k\leq j} c_k.
\end{equ}
Let $Z_{ij}$ be an $n_{\delta(i)}\times n$  matrix all of whose entries are zero of size when $\delta(i)\neq j$ and $B_{\delta(i)}$ if $\delta(i)=j$, for $i,j\in[m']$.
Define $nc \times n$ matrices $T_j$ for $j\in[m]$ as follows:
\begin{equ}
T_j := \begin{bmatrix}
Z_{1j}\\
\vdots\\
Z_{m'j}
\end{bmatrix},
\end{equ}
and define the linear operator $T_{\B,\p}:\Sn^{nc}_{++}\to\Sn^{n}_{++}$ as 
\begin{equ}[eq:OS-operator]
T_{\B,\p}(X) := \sum_{j\in[m']} T_j^\top X T_j.
\end{equ}

\begin{theorem}[Reduction from Brascamp-Lieb to Operator Scaling\cite{garg2016algorithmic}, Lemma 4.4.]
Let $(\B,\,\p)$ be a Brascamp-Lieb a datum with $B_j\in\R^{n_j\times n}$ and $p_j =c_j/c$ where $c,c_j\in\Z_+$ for each $j\in[m]$.
 The capacity of the operator $T_{\B,\p}$ defined in~\eqref{eq:OS-operator} satisfies $\mathrm{cap}(T_{\B,\p})=\nicefrac{1}{\BL(\B,\p)^2}$.
\end{theorem}

%
\noindent
While this gives a geodesically log-concave formulation to compute the Brascamp-Lieb constant, the algorithm of Allen-Zhu {\em et al.}~\cite{OS2018} is not enough to obtain a polynomial time algorithm to compute the capacity of this operator as the dimension of the optimization problem is exponentially large in the bit complexity of the Brascamp-Lieb datum.

\section{Proof of \cref{thm:Lieb}}
\label{sec:original_formulation}

Let $(\B,\p)$ be a feasible Brascamp-Lieb datum with $B_j\in\R^{n_j\times n}$ and $A:=(A_j)_{j\in[m]}$ with $A_j\in\Sn^{n_j}_{+}$ be the input of $\BL(\B,\p;A)$ as defined in~\eqref{eq:BLconstant}.
To prove the joint geodesic convexity of $\BL(\B,\p;A)$ with respect to $A$ we extend~\cref{thm:log-lm-gc} and~\cref{thm:plm-on-gm} from linear maps to ``jointly linear maps''.
We use the term jointly linear maps to refer to multivariable functions of the form $\sum_{j\in[m]} \Phi_j(P_j)$ where each $\Phi_j$ is a strictly positive linear map for each $j\in[m]$.
In particular, the term $\sum_{j\in[m]} p_j B_j^\top A_j B_j$ in~\eqref{eq:BLconstant} is a jointly linear map.

The extension of~\cref{thm:plm-on-gm} is presented in~\cref{thm:jlm-on-gm} and its proof is based on the maximal characterization of geometric mean (\cref{thm:gm-maximal}) and the effect of positive linear maps on the positive definiteness of block matrices (\cref{thm:plm-effect}).
We follow the proof of~\cref{thm:plm-on-gm} for each $\Phi_j$, but instead of concluding $\Phi_j(P_j\#Q_j)\preceq \Phi_j(P_j)\#\Phi_j(Q_j)$ from the maximality of geometric mean, we sum the resulting inequalities.
Subsequently,~\cref{thm:jlm-on-gm} follows from the maximality of geometric mean.
\cref{lem:log-jlm-gc} is an extension of ~\cref{thm:log-lm-gc} and follows directly from~\cref{thm:jlm-on-gm}.

\begin{definition}[Jointly linear map]
Let $\Phi:\Sn^{n_1}_{+}\times\cdots\times\Sn^{n_m}_{+}\to\Sn^n_{+}$.
We say that $\Phi$ is a jointly linear map if there exist strictly positive linear maps  $\Phi_j:\Sn^{n_j}_{+}\to\Sn^n_{+}$ such that 
\begin{equ}[eq:jointly-linear]
\Phi(P_1,\hdots,P_k) := \sum_{j\in[k]} \Phi_j(P_j).
\end{equ}
\end{definition}

\noindent
Now, we state the extension of~\cref{thm:plm-on-gm}.
%

\begin{theorem}[Effect of Jointly Linear Maps over Geometric Means]
Let $\Phi:\Sn^{n_1}_{+}\times \cdots\times \Sn^{n_m}_{+}\to \Sn^n_{+}$ be a jointly linear map.
Then, 
\begin{equ} 
\Phi(G) \preceq \Phi(P) \# \Phi(Q) 
\end{equ}
where $P:=(P_j)_{j\in[m]}$, $Q:=(Q_j)_{j\in[m]}$, and $G:=(G_j)_{j\in[m]}$ with $P_j,Q_j\in\Sn^{n_j}_{++}$ and $G_j := P_j \# Q_j$.
\label{thm:jlm-on-gm}
\end{theorem}

\noindent
The following is a corollary of the theorem above and a generalization of~\cref{thm:log-lm-gc}.

\begin{corollary}[Joint Geodesic Convexity of Logarithm of  Jointly Linear Maps]
If $\Phi:\Sn^{n_1}_{+}\times\cdots\times\Sn^{n_m}_{+}\to\Sn^n_{+}$ is a jointly linear map, then 
\begin{equ}[eq:c0]
g(P_1,\hdots,P_m):=\log\det(\Phi(P_1,\hdots,P_m))
\end{equ}
is jointly geodesically convex in $\Sn^{n_1}_{++}\times\cdots\times\Sn^{n_m}_{++}$.
\label{lem:log-jlm-gc}
\end{corollary}

\begin{proof}[Proof of~\cref{lem:log-jlm-gc}]
We show that $g$ is jointly geodesically mid-point concave.
\cref{thm:jlm-on-gm} implies that 
\begin{equ}[eq:c1]
\Phi(G)\preceq\Phi(P)\#\Phi(Q)
\end{equ}
for any $P:=(P_j)_{j\in[m]}$ and $Q:=(Q_j)_{j\in[m]}$ with $P_j,Q_j\in\Sn^{n_j}_{+}$.
Therefore,
\begin{align*}
g(G)
	\stackrel{\eqref{eq:c0}}{=}& \log\det(\Phi(G))\\
	\leq & \log\det(\Phi(P)\#\Phi(Q))&(\text{monotonicity of $\log\det$ and~\eqref{eq:c1}})\\
	\leq & \frac{1}{2}\log\det(\Phi(P))+\frac{1}{2}\log\det(\Phi(Q))&(\text{multiplicativity of $\det$})\\
	\stackrel{\eqref{eq:c0}}{=} & \frac{1}{2}\left(g(P)+g(Q)\right).
\end{align*}
Thus, $g$ satisfies mid-point geodesic convexity.
Consequently, we establish the geodesic convexity of $g$ using the continuity of $g$.
\end{proof}

\noindent
The proof of~\cref{thm:Lieb} is a simple application of~\cref{lem:log-jlm-gc} and~\cref{thm:det-linear}.

\begin{proof}[Proof of~\cref{thm:Lieb}]
We show that $\BL(\B,\p;A)$ is jointly geodesically mid-point log-concave with respect to $A$.
In other words, we show that for arbitrary $P=(P_j)_{j\in[m]}$, $Q=(Q_j)_{j\in[m]}$
\begin{equ}
-\log \BL(\B,\p;G) \leq -\frac{1}{2}\left(\log \BL(\B,\p;P)+\log \BL(\B,\p;Q)\right)
\end{equ}
 where $G=(G_j)_{j\in[m]}$ with $G_j:=P_j\#Q_j$, being the midpoint of geodesic combining $P_j$ to $Q_j$.
This implies that $\BL(\B,\p;A)$ is jointly geodesically log-concave with respect to $A$ due to the continuity of $\BL(\B,\p;A)$ with respect to $A$.

Let $\Phi_j(P_j) := p_j B_j^\top P_j B_j$.
$\Phi_j$ is strictly positive linear map by~\cref{lem:feasible-plm}.
Then, $\Phi(P) := \sum\limits_{j\in[m]} p_j B_j^\top P_j B_j$ is jointly linear, as $\Phi(P) = \sum\limits_{j\in[m]} \Phi_j(P_j)$.
Hence, $\log\det(\Phi(P))$ is jointly geodesically convex by~\cref{lem:log-jlm-gc}.
Also, $\log\det(X)$ is geodesically linear (\cref{thm:det-linear}).
Thus, for any $P:=(P_j)_{j\in[m]}$, $Q:=(Q_j)_{j\in[m]}$ and $G:=(G_j)_{j\in[m]}$ with $G_j := P_j\# Q_j$ we have
\begin{align*}
-\log \BL(\B,\p;G)
	\stackrel{\eqref{eq:BLconstant}}{=}& \frac{1}{2}\left(\log\det(\Phi(G)) - \sum_{j\in[m]} p_j \log\det(G_j)\right)\\
	\stackrel{\eqref{eq:c0}}{\leq}& \frac{1}{2}\left(\frac{1}{2}\left(\log\det(\Phi(P))+\log\det(\Phi(Q))\right) - \sum_{j\in[m]} p_j \log\det(G_j)\right)\\
	=&\frac{1}{2}\left(\frac{1}{2}\left(\log\det(\Phi(P))-\sum_{j\in[m]}p_j\log\det(P_j)\right)\right.\\
	&\qquad\qquad\left.+\frac{1}{2}\left(\log\det(\Phi(Q))-\sum_{j\in[m]}p_j\log\det(Q_j)\right)\right)&(\text{\cref{thm:det-linear}})\\
	\stackrel{\eqref{eq:BLconstant}}{=}& -\frac{1}{2}\left(\log\det\BL(\B,\p;P)+\log\det\BL(\B,\p;Q)\right).
\end{align*}  
This concludes the proof.
\end{proof}

\noindent
Now we prove~\cref{thm:jlm-on-gm}.
This proof is based on the proof of~\cref{thm:plm-on-gm} and depends on the maximality of geometric mean (\cref{thm:gm-maximal}) and effects of positive linear maps on block matrices (\cref{thm:plm-effect}).

\begin{proof}[Proof of~\cref{thm:jlm-on-gm}]
$\Phi$ is a jointly linear map by the assumption.
Thus, there exist linear maps $\Phi_j:\Sn^{n_j}_{+}\to \Sn^n_{+}$ such that $\Phi(P) = \sum_{j\in[m]} \Phi_j(P_j)$.
\cref{thm:gm-maximal} implies for each $j\in[m]$,
\begin{equ}
0\preceq
\begin{bmatrix}
P_j&G_j\\
G_j&Q_j
\end{bmatrix}.
\end{equ}
Since $\Phi_j$'s are strictly positive linear maps,~\cref{thm:plm-effect} implies that for each $j\in[m]$,
\begin{equ}
0\preceq
\begin{bmatrix}
\Phi_j(P_j)&\Phi_j(G_j)\\
\Phi_j(G_j)&\Phi_j(Q_j)
\end{bmatrix}.
\end{equ}
The dimension of these block matrices is $2n\times 2n$ for each $j\in[m]$.
Thus we can sum these inequalities and the summation leads to
\begin{equs}
0
	\preceq& \sum_{j\in[m]}
\begin{bmatrix}
\Phi_j(P_j)  & \Phi_j(G_j)\\
\Phi_j(G_j) &  \Phi_j(Q_j)
\end{bmatrix}\\
	=&
\begin{bmatrix}
\sum\limits_{j\in[m]}(\Phi_j(P_j)  	& \sum\limits_{j\in[m]} \Phi_j(G_j)\\
\sum\limits_{j\in[m]} \Phi_j(G_j)	& \sum\limits_{j\in[m]} \Phi_j(Q_j) 
\end{bmatrix}\\
	\stackrel{\eqref{eq:jointly-linear}}{=}&
\begin{bmatrix}
\Phi(P) &\Phi(G)\\
\Phi(G) &\Phi(Q)
\end{bmatrix}.\label{eq:joint-ineq}
\end{equs}
\cref{thm:gm-maximal} and~\eqref{eq:joint-ineq} imply that $\Phi(G) \preceq \Phi(P)\#\Phi(Q)$.
\end{proof}

\section{Proof of~\cref{thm:new_formulation}}
\label{sec:new_formulation}

Let $(\B,\p)$ be a feasible Brascamp-Lieb datum with $B_j\in\R^{n_j\times n}$.
Let $A:=(A_j)_{j\in[m]}$ with $A_j\in\Sn^{n_j}_{+}$ be the input of $\BL(\B,\p;A)$ as defined in~\eqref{eq:BLconstant}.
The proof of~\cref{thm:new_formulation} first establishes the geodesic concavity of $F_{\B,\p}$ as defined in~\eqref{eq:F} when $(\B,\p)$ is feasible.
Next, it establishes the relation between global maximizers of $F_{\B,\p}(X)$ and global maximizers of $\BL(\B,\p;A)$, as well as the relation between $\sup_{X\in\Sn^{n}_{++}} F_{\B,\p}(X)$ and $\BL(\B,\p)$ when $(\B,\p)$ is simple.
Feasibility of $(\B,\p)$ implies the linear maps $B_j X B_j^\top$ are strictly positive linear for each $j\in[m]$ (\cref{lem:feasible-plm}).
Consequently, $-\log\det(B_j X B_j^\top)$ is geodesically concave by~\cref{thm:log-lm-gc} for each $j\in[m]$.
Also, $\log\det(X)$ is geodesically concave by~\cref{thm:det-linear}.
Thus, $F_{\B,\p}(X)$ is geodesically concave as a sum of geodesically concave functions with non-negative coefficients.

The geodesic concavity of $F_{\B,\p}$ implies that any local maximum is also a global maximum (\cref{thm:gc-minimizers}).
Consequently, we investigate the points where all directional derivatives of $F_{\B,\p}$ vanish, the critical points of $F_{\B,\p}$.
A simple calculation involving the first derivative shows that any critical point $X$ of $F_{\B,\p}$ should satisfy
\begin{equ}
X^{-1} = \sum_{j\in[m]} p_j B_j^\top (B_j X B_j)^{-1} B_j.
\end{equ}
\cref{thm:EquivalencetoGeometric} implies that we can construct a global maximizer of $\BL(\B,\p;A)$  from $X$ by setting $A_j:=(B_j XB_j)^{-1}$.
Furthermore, we can construct a critical point of $F_{\B,\p}$ using a global maximizer of $\BL(\B,\p;A)$ by setting $X:=(\sum_{j\in[m]} p_j B_j^\top A_j B_j)^{-1}$.
\cref{thm:EquivalencetoGeometric} guarantees the existence of a global maximizer of $\BL(\B,\p;A)$ if $(\B,\p)$ is simple.
Thus, if $(\B,\p)$ is simple, then $\sup_X F_{\B,\p}(X)$ should be attained.
We can deduce $\sup_X F_{\B,\p}(X)=2\log\BL(\B,\p)$ from the construction of $F_{\B,\p}$ and the relation between maximizers of $F_{\B,\p}$ and $\BL(\B,\p;A)$.

The second part of the proof~\cref{thm:new_formulation} depends on well-known identities from matrix calculus.
We present these identities for the convenience of the reader, and refer the interested reader to the matrix cookbook~\cite{petersen2008matrix} for more details.

\begin{proposition}
Let $X(t), Y(t)$ be a differentiable function from $\R$ to $d\times d$ invertible symmetric matrices.
Let $U\in \R^{d'\times d}$, $V\in R^{d\times d''}$, $W\in\R^{d\times d}$ be matrices which do not depend on $t$.
Then, the following identities hold:
\begin{equs}
\frac{d \log\det(X(t))}{dt} 	=& \tr\left(X(t)^{-1} \frac{d X(t)}{dt}\right)\label{eq:d-logdet}\\
\frac{d UX(t)V}{dt}		=& U\frac{d X(t)}{dt} V\label{eq:d-conj}\\
\frac{d tW}{dt}			=& W.\label{eq:d-linear}\\
\end{equs}
\end{proposition}

\begin{proof}[Proof of~\cref{thm:new_formulation}]
We start by showing that $F_{\B,\p}(X)$ is geodesically concave.
The feasibility of $(\B,\p)$ implies $B_j X B_j^\top$ is a strictly positive linear map for each $j\in[m]$ (\cref{lem:feasible-plm}).
Thus,~\cref{thm:log-lm-gc} yields that for any $X,Y\in\Sn^{n}_{++}$,
\begin{equ}[eq:g-convex-0]
\log\det(B_j (X\#Y) B_j^\top) \leq \frac{1}{2}\log\det(B_j X B_j^\top) + \frac{1}{2}\log\det(B_j Y B_j^\top).
\end{equ}
Combining this with the geodesic linearity of $\log\det(X)$ (\cref{thm:det-linear}), we obtain
\begin{equs}
\log\det(X\#Y) - \sum_{j\in[m]} p_j \log\det(B_j (X\#Y) B_j^\top) \geq &
\frac{1}{2}\left(\log\det(X) - \sum_{j\in[m]} p_j \log\det(B_j X B_j^\top)\right)\\
&\qquad+
\frac{1}{2}\left(\log\det(Y) - \sum_{j\in[m]} p_j \log\det(B_j Y B_j^\top)\right).
\end{equs}
Equivalently,
\begin{equ}
F_{\B,\p}(X\#Y)\geq \frac{1}{2}F_{\B,\p}(X)+\frac{1}{2}F_{\B,\p}(Y).
\end{equ}
Therefore, $F_{\B,\p}(X)$ is geodesically concave.

Now, we can show the second part of the theorem.
The geodesic concavity of $F_{\B,\p}$ implies any local maximum of $F_{\B,\p}$ is a global maximum of $F_{\B,\p}$.
A local maximum of $F_{\B,\p}$ is achieved at $X$ if it is a critical point of $F_{\B,\p}$.
If $X$ is a critical point of $F_{\B,\p}$, then for any symmetric matrix $Q$, the directional derivative of $F_{\B,p}$ at $X$ in the direction of $Q$ should be 0.
In other words, if $\zeta(t):=X+tQ$ and $f(t):=F_{\B,\p}(\zeta(t))$, then $\left.\frac{df}{dt}\right\rvert_{t=0}$ should be 0 for any $Q$.
Let us compute $\frac{df}{d t}$,
\begin{equs}
\frac{df}{dt}
	\stackrel{\eqref{eq:F}}{=}& \frac{d}{dt}\log\det(\zeta(t)) -\sum_{j\in[m]} p_j \frac{d}{dt}\log\det(B_j \zeta(t) B_j^\top)\\
	\stackrel{\eqref{eq:d-logdet}}{=}& \tr\left(\zeta(t)^{-1}\frac{d \zeta(t)}{dt}\right) - \sum_{j\in[m]} p_j \tr\left((B_j \zeta(t) B_j^\top)^{-1} \frac{d B_j \zeta(t) B_j^\top}{dt}\right)\\
	\stackrel{\eqref{eq:d-conj}}{=}& \tr\left(\zeta(t)^{-1}\frac{d \zeta(t)}{dt}\right) - \sum_{j\in[m]} p_j \tr\left((B_j \zeta(t) B_j^\top)^{-1} B_j\frac{d  \zeta(t)}{dt}B_j^\top\right)\\
	\stackrel{\eqref{eq:d-linear}}{=}& \tr(\zeta(t)^{-1}Q) - \sum_{j\in[m]} p_j\tr((B_j\zeta(t) B_j^\top)^{-1} B_j Q B_j^\top).
\end{equs}
Hence, the directional derivative of $F_{\B,\p}(X)$ in the direction of $Q$, $\left.\frac{df}{dt}\right\rvert_{t=0}$ is
\begin{equ}[eq:d-derivative]
\tr(X^{-1}Q)-\sum_{j\in[m]} p_j\tr((B_j X B_j^\top)^{-1} B_j Q B_j^\top).
\end{equ}
If directional derivates of $F_{\B,\p}$ vanish at $X$, then~\eqref{eq:d-derivative} should be 0 for any symmetric matrix $Q$.
Consequently,
\begin{equ}
\tr(Q[X^{-1}-\sum_{j\in[m]} p_j B_j^\top (B_j X B_j^\top)^{-1} B_j]) = 0.
\end{equ} 
This observation leads to
\begin{equ}[eq:critical-point]
X^{-1} = \sum_{j\in[m]} p_j B_j^\top (B_j X B_j^\top)^{-1} B_j.
\end{equ}
If $(\B,\p)$ is simple, then there exists an input $A^\star=(A_j)_{j\in[m]}$ such that $A_j^{-1} = B_j M^{-1} B_j^\top$ where $M=\sum_{j\in[m]} p_j B_j^\top A_j B_j$ by~\cref{thm:EquivalencetoGeometric}.
Consequently, $M$ satisfies
\begin{equ}[eq:identity]
M = \sum_{j\in[m]} p_j B_j^\top (B_j M^{-1} B_j^\top)^{-1} B_j.
\end{equ}
If $X^\star:=M^{-1}$, then $X^\star$ satisfies~\eqref{eq:critical-point} due to~\eqref{eq:identity}.
Thus, $X^\star$ is a critical point of $F_{\B,\p}(X)$ and $F_{\B,\p}$ attains its maximal value at $X^\star$.
Furthermore, the maximizer $A^\star$ of $\BL(\B,\p;A)$ is equal to $((B_j X^\star B_j^\top)^{-1})_{j\in[m]}$.
Finally,
\begin{equs}
F_{\B,\p}(X^\star) 
	\stackrel{\eqref{eq:F}}{=}& \log\det(X^\star) - \sum_{j\in[m]} p_j \log\det(B_j X^\star B_j^\top) \\
	=& \log\det\left(\left(\sum_{j\in[m]}p_j B_j^\top(B_jX^\star B_j^\top)^{-1}B_j\right)^{-1}\right) - \sum_{j\in[m]} p_j \log\det(B_j X^\star B_j^\top) \\
	=& \log\det\left(\left(\sum_{j\in[m]} p_j B_j^\top A_j B_j\right)^{-1}\right) - \sum_{j\in[m]} p_j\log\det(A_j^{-1})\\
	=& \sum_{j\in[m]} p_j \log\det(A_j)-\log\det\left(\sum_{j\in[m]} p_j B_j^\top A_j B_j\right).
\end{equs}
Therefore, $\BL(\B,\p;A^\star) = \exp(\frac{1}{2} F_{\B,\p}(X^\star))$.

\end{proof}

\section*{Acknowledgment}
We thank Ankit Garg for suggesting the possibility of  geodesic log-concavity of Lieb's  formulation.

\bibliographystyle{alpha}
\bibliography{refs} 

\end{document}